\documentclass[11pt]{article}
\usepackage{amssymb}
\usepackage{array}
\usepackage{amsmath}
\usepackage{color}
\usepackage[colorlinks,linkcolor=blue,citecolor=green]{hyperref}
\newtheorem{theorem}{Theorem}

\newtheorem{corollary}[theorem]{Corollary}

\newtheorem{example}[theorem]{Example}

\newtheorem{remark}[theorem]{Remark}
\newenvironment{proof}[1][Proof]{\textbf{#1.} }{\ \rule{0.5em}{0.5em}}

\usepackage{graphicx}
\allowdisplaybreaks[3]
\usepackage{graphicx} 

\title{\Large \textbf{Nonlocal Hamiltonian structures  of  the kinetic equation for soliton gas under polychromatic reductions }}

\author{Pierandrea Vergallo$^{1,2}$\\
\small $^{1}$Department of Engineering, University of Messina,\\
\small \texttt{pierandrea.vergallo@unime.it}\\
\small $^{2}${Istituto Nazionale di Fisica Nucleare, Sez.\ Milano}}

\begin{document}

\maketitle

\vspace{8mm}

\abstract{We deepen the existence of a nonlocal Hamiltonian formalism for the El's kinetic equation for soliton gas under the polychromatic reduction for a class of interaction kernels. The nonlocality presented is related to semi-Riemannian metrics of constant curvature, conformally flat metrics and hypersurfaces in a pseudo-Euclidean space. These results generalise a previous one that Vergallo and Ferapontov obtained with local Hamiltonian operators. Some examples as the Korteweg-de Vries, the Lieb-Liniger and the separable cases are analysed.  }

\vspace{20mm}

\tableofcontents

\section{Introduction}
The kinetic equation for soliton gas is an integro-differential equation that recently appeared in the context of the soliton gas theory \cite{CER,El,EK} and, independently, in what is known as generalised hydrodynamics \cite{Do,Do1,Do2} (for this reason it is also known as GHD equation). Both the theories deal with emerging phenomena in mathematical physics (making use of statistical mechanics, integrable systems and nonlinear waves) and experimental physics (due to their wide applications in the fields of optics, hydrodynamics and water waves). The main equations (also indicated by El's equation especially among the community of soliton gas theory) describe an infinite stochastic ensamble of interacting solitons with randomly distributed parameters. Moreover, in \cite{El}, the author proved that such an equation can be obtained as thermodynamic limit of the Whitham equations for the famous Korteweg-de Vries equation. This result was then generalised for a number of very well-known integrable systems (the sinh-Gordon equation, the Toda lattice, the defocusing nonlinear Schroedinger equation..), connecting directly the analytical and geometric approaches used in the theory of integrable systems to the emerging phenomena. This connection led researchers to investigate the geometric inner structure of the equations, focusing in particular on their Hamiltonian formulation. 

After the introduction of the delta-functional reduction (see \cite{EKPZ,PTE}), it has been proved that the integro-differential structure simply reduces to a quasilinear system of first-order Partial Differential Equations (PDEs in what follows). After that, its integrability was proved in \cite{FP}, mapping the system previously obtained into one of Jordan-block type. Recently, in \cite{VerFer1,VerFer2} the Hamiltonian structure of these reduced systems has been proved, using the Hamiltonian formalism for PDEs  with local differential operators, known as Dubrovin-Novikov operators. In \cite{VerFer2}, the authors also stated that a second Hamiltonian structure can be investigated, which is nonlocal and that deals with a compatible pair of operators. As proved by Magri \cite{magri}, such a bi-Hamiltonian structure leads to the integrability of the system. 

Finally, a challenging topic in the recent papers \cite{thib} and \cite{VerFer2} was the investigation of a Hamiltonian structure for the full  integro-differential kinetic equation, without any reduction. {In \cite{buld}, Bulchandani studied the case of nonlocal Hamiltonian structures for this equation. Even if the reference \cite{buld} is mostly devoted to quantum integrable systems, the structures obtained are very closed to the ones we investigate here and the paper represents an interesting discussion on this subject.   However, a systematic approach to Poisson structures for integro-differential systems (and consequently to kinetic equations in general) remains an open problem. We hope this paper can help in this direction starting from the investigations of nonlocalities for the polychromatic reduction of the system, i.e. under the delta-functional ansatz. }

{The paper is structured as follows: in Section \ref{2} we introduce the mathematical description of the equation, its polychromatic reduction and we briefly review the main concepts of (local) Hamiltonian formalism for PDEs. In Section \ref{sec0} we recall the structure of nonlocal operators and we prove the main Theorem of the paper, presenting some related Corollaries. Finally in Section \ref{sub33} we show some applications of the obtained structures for integrable systems. At the end, we discuss some further perspectives in the Conclusions in Section \ref{conc}.}

\section{Polychromatic reductions of the El's equation}\label{2}

 Let us consider the El's integro-differential kinetic equation, that  describes the evolution of a dense soliton gas and is also considered a generalisation of Zakharov's kinetic equation for rarefied soliton gas \cite{Z}. The equation concretely presents as a pair of equations, the first one is a hydrodynamic-type system whereas the second is an integral equation:
\begin{equation}\label{gas}
\begin{array}{c}
f_t+(sf)_x=0,\\
\ \\
{\displaystyle s(\eta)=S(\eta)+\mathop{\int} G(\mu, \eta)f(\mu)[s(\mu)-s(\eta)]\ d\mu,}
\end{array}
\end{equation}
where 
\begin{itemize}
    \item $f(\eta)=f(\eta, x, t)$ is the distribution function describing the density of the soliton gas (the so called density of states) such that {$f(\eta_0,x_0,t_0)\, d\eta \, dx$ is the number of solitons found at $t=t_0$ in the element $[\eta_0,\eta_0+d    \eta]\times [x_0,x_0+dx]$ of the phase space}. Note that we integrate over the support of $f$; \item $s(\eta)=s(\eta, x, t)$ is the associated transport velocity (also known as effective transport velocity);  
    \item $\eta$ is a spectral parameter in the Lax pair associated with the dispersive hydrodynamics; 
    \item $S(\eta)$ is the free soliton velocity and describes the velocity of the gas in the vacuum and \item $G(\mu, \eta)$ is the interaction kernel (also known as phase shift) modelling the pairwise soliton collisions.  We require $G(\mu, \eta)$  to be symmetric in its entries, i.e. $$G(\mu, \eta)=G(\eta, \mu).$$ 
\end{itemize}

The latter functions $S(\eta)$ and $G(\mu,\eta)$ are independent of $x$ and $t$.

We now consider the delta-functional ansatz for the distribution function $f(\eta,x,t)$:
\begin{equation}\label{del}
 f(\eta, x, t)=\sum_{i=1}^{n}u^i(x, t)\ \delta(\eta-\eta^i(x, t)),
\end{equation}
where $u^i(x,t)$ physically represent the propagation weights according to their transport velocities $s_i=s(\eta^i,x,t)$. By substituting \eqref{del} into \eqref{gas} the kinetic equation is reduced to a quasilinear system of $2n$ equations  
\begin{equation}\label{sysred1}
    \begin{cases}
        u^i_t=\partial_x(v^iu^i)\\
        \eta^i_t=v^i\eta^i_x
    \end{cases},
\end{equation}
in the $2n$ field variables $u^i(x,t)$ and $\eta^i(x,t)$, $i=1,2,\dots n$. This result was firstly derived by Pavlov, Taranov and El in 2012 in \cite{PTE} and the reduction is known as delta-functional ansatz or polychromatic reduction. Here the coefficients $v^i$ coincide up to a sign with the effective velocities of the weights $u^i$,  $v^i=-s(\eta^i,x,t)$, and can be recovered by the following system (linear in $v^j$):
\begin{equation}
    v^i(u,\eta)=-S(\eta^i)+\sum_{k\ne i}\epsilon^{ki}u^k(v^k-v^i), \quad  \text{with}\quad  \epsilon^{ki}={G(\eta^k, \eta^i)}, \ k\ne i.
\end{equation}

Let us now assume $\eta^i(x,t)$ are not constant, otherwise the last $n$ equations in \eqref{sysred1} are trivially satisfied, as investigated in \cite{EKPZ}. First, we recall the transformation of the field variables $(u,\eta)\mapsto (r,\eta)$, where
\begin{equation}
    r^i=-\frac{1}{u^i}\left(1+\sum_{k\ne i}\epsilon^{ki}u^k \right), \qquad i=1,2,\dots n,
\end{equation} firstly introduced by Ferapontov and Pavlov in \cite{FP}. Under the previous assumption on $\eta^i$, the system is mapped into 

\begin{equation}
\label{J1}
\begin{cases}
r^i_t=v^i r^i_x+p^i\eta^i_x  \\
\eta^i_t=v^i\eta^i_x  
\end{cases}
\end{equation}
whose coefficients can be explicitly expressed after the following definitions. Let us introduce a symmetric $n\times n$ matrix $\hat{\epsilon}$ whose entries are 
\begin{equation}
    \hat{\epsilon}^{ij}=\begin{cases}
        \epsilon^{ij}\qquad i\neq j\\
        r^i\qquad i=j 
    \end{cases}.
\end{equation}In addition, we set $\beta:=-\hat{\epsilon}^{-1}$, whose entries are indicated by $\beta_{ij}$. Finally, we obtain the following formulas
\begin{align}\begin{split}\label{uvp}
u^i&=\sum_{k=1}^n\beta_{ki}, \\  v^i&=\frac{1}{u^i}\sum_{k=1}^n \beta_{ki}\xi^k, \\
p^i&=\frac{1}{u^i}\left(\sum_{k=1}^n \epsilon^{ki}_{,\eta^i}(v^k-v^i)u^k+(\xi^i)' \right),
\end{split}\end{align}
where {$\xi^i(\eta^i)=-S(\eta^i) $ and }we use the notation $\epsilon^{ki}_{,\eta^i}$ to indicate partial derivative with respect to  $\eta^i$.

\begin{remark}
    The resulting system is  composed by {$n$ blocks of size $2\times 2$} in the diagonal. Each block is upper triangular, in Toeplitz form. Systems of this type are also known as parabolic hydrodynamic-type systems or systems of Jordan-block type. A recent investigation of parabolic system within the context of finding solutions via the method of differential constraints was done in \cite{RizVer1}. 
\end{remark}

\subsection{Hamiltonian structures of the reduced system}\label{3}
In this {subsection}, we describe the general framework of Hamiltonian formalism for PDEs. In particular, we focus on an evolutionary system of {\color{blue}$n$} equations
\begin{equation}\label{evsys}
    {u^i_t=f^i(x,\textbf{u},\dots,\textbf{u}_{kx}), \qquad i=1,2,\dots n, }
\end{equation}
where $t,x$ are the independent variables and $\textbf{u}=(u^1,\dots, u^n)$ are $n$ dependent (or field) variables depending on $t$ and $x$, here $u_{hx}$ stays for $\frac{\partial^h u}{\partial x^h}$. In this context, we say that the maximum order of derivation appearing in the system is the order of the system. 

The Hamiltonian formalism for evolutionary systems of PDEs \eqref{evsys} has become a well-studied topic in the last fifty years. Finding a Hamiltonian structure for a given system reveals additional properties  such as the existence of conserved quantities or a deeper geometric structure behind its analytic representation. The key role of the Poisson tensor and the Hamiltonian function for ODEs is covered by the Hamiltonian operator and the Hamiltonian functional respectively. We briefly recall that a matrix differential operator $\mathcal{A}^{ij}$, with $i,j=1,2,\dots n$ is Hamiltonian if it is skew-adjoint ($\mathcal{A}^*=-\mathcal{A}$) and its Schouten {bracket vanishes} ($[\mathcal{A},\mathcal{A}]=0$). For those who are more familiar with Poisson brackets, the previous requirements are equivalent to the fact that for any functionals $F=\int{f\, dx}$, {\color{blue} $G=\int{g\, dx}$} the bracket defined with respect to the operator $\mathcal{A}$
\begin{equation}
    \{F,G\}_{\mathcal{A}}=\displaystyle \int{\frac{\delta f}{\delta u^i}\mathcal{A}^{ij}\left(\frac{\delta g}{\delta u^j}\right)\, dx}
\end{equation}is a Poisson {bracket}, i.e. it is skew-symmetric ($\{F,G\}_{\mathcal{A}}=-\{G,F\}_{\mathcal{A}}$) and for any additional functional $H=\int{h\, dx}$ satisfies the Jacobi identity $$\{\{F,G\}_{\mathcal{A}},H\}_{\mathcal{A}}+\{\{G,H\}_{\mathcal{A}},F\}_{\mathcal{A}}+\{\{H,F\}_{\mathcal{A}},G\}_{\mathcal{A}}=0.$$
Finally, a Hamiltonian functional is a functional 
\begin{equation}
    H=\displaystyle \int{h(x,\textbf{u},\dots, \textbf{u}_{mx})\, dx},
\end{equation}
such that a system of the form \eqref{evsys} is written as 
\begin{equation}
    {u^i_t=\{u^i,H\}_{\mathcal{A}}, \qquad i=1,2,\dots n}
\end{equation} for a specific choice of a Hamiltonian operator $\mathcal{A}$. Note that the previous expression explicitly reads as
\begin{equation}
    {u^i_t=\mathcal{A}^{ij}\left(\frac{\delta H}{\delta u^j}\right), \qquad i=1,2,\dots n.}
\end{equation}

Several scalar equations and systems in physical phenomena have been proved to be Hamiltonian (e.g. the Korteweg-de Vries, the non-linear Schroedinger, the KP system, the Camassa-Holm, the sin-Gordon equation..). We refer to the survey \cite{mokhov98:_sympl_poiss} for further details and for a larger number of examples.

\paragraph{{Dubrovin-Novikov operators}}
As shown {at the beginning of this Section}, in this paper we focus on systems of PDEs which are quasilinear of first order and homogeneous, i.e. 
\begin{equation}\label{hts}
    u^i_t=V^i_j(u)u^j_x, \qquad i=1,2,\dots n.
\end{equation}
Systems of these type are also known in the literature as hydrodynamic-type systems \cite{Tsarev,Tsarev1}. 

In 1983, Dubrovin and Novikov proved that a natural {\color{blue}local} Hamiltonian structure exists for a large class of hydrodynamic-type systems. Indeed, one can assume the Hamiltonian density to  simply depend on the field variables $\textbf{u}$ and not on their higher order derivatives. Such functionals are also known as \emph{hydrodynamic functionals}, and to investigate a first-order system the only possibility is that the operator is a differential operator of degree 1. The most general structure for first-order operators which are homogeneous in the degree of derivation\footnote{Here we use the natural grading rules, i.e. $$\text{deg}(\partial_x^k)=k, \qquad \text{deg}(u_{hx})=h.$$} is the following
\begin{equation}\label{dn}
    g^{ij}(u)\partial_x+b^{ij}_k(u)u^k_x, \qquad i,j=1,2,\dots n.
\end{equation}
Here the coefficients $g^{ij},b^{ij}_k$ depend on the field variables only (so that their degree is zero). A central role in these operators is played by the leading coefficient $g^{ij}$, indeed in the non-degenerate case $\det(g)\neq 0$ the following result holds true:
\begin{theorem}[\cite{DN83}]\label{dnthm}
    A first-order homogeneous operator of form \eqref{dn} is a Hamiltonian operator if and only if setting $g_{ij}=(g^{kl})^{-1}$:
    \begin{itemize}
        \item $g_{ij}$ is a flat metric; 
        \item $b^{ij}_k$ and the Christoffel symbols for $g$ are related by the formula
        \begin{equation}
            b^{ij}_k=-g^{is}\Gamma^j_{sk}, \qquad i,j,k=1,2,\dots n.
        \end{equation}
    \end{itemize}
\end{theorem}
First-order homogeneous Hamiltonian operators are also known as \emph{Dubrovin-Novikov operators}.  In what follows we assume $\det(g)\neq 0$ so that we substitute $b^{ij}_k$ more intuitively with $\Gamma^{ij}_k$. 

Now, a given hydrodynamic-type system \eqref{hts} admits a Hamiltonian structure with a Dubrovin-Novikov operator if a hydrodynamic functional $h(\textbf{u})$ exists such that 
\begin{equation}
    u^i_t=V^i_j(u)u^j_x=\left(\nabla^i\nabla_j h\right)u^j_x, \qquad i=1,2,\dots n,
\end{equation}
where $\nabla_j$ is the covariant derivative with respect to the Levi-Civita connection of $g_{ij}$ and $\nabla^i=g^{is}\nabla_j$.

A useful result to prove the Hamiltonianity of hydrodynamic-type systems was presented by Tsarev:

\begin{theorem}[\cite{Tsarev,Tsarev1}]\label{tsathm}
    A system of hydrodynamic type \eqref{hts} is Hamiltonian with a Dubrovin-Novikov structure if and only if the following relations are satisfied
    \begin{subequations}\begin{align}
        &g^{is}V^j_s=g^{js}V^i_s, \qquad i,j=1,2,\dots n,\label{cond1}\\
        &\nabla^iV^j_k=\nabla^jV^i_k, \qquad i,j,k=1,2,\dots n.\label{cond2}
    \end{align}
    \end{subequations}
\end{theorem}

\vspace{3mm}

{After presenting the general framework in which we operate, we can now focus on the investigated equation under polychromatic reduction 
\begin{equation}
\label{J2}
\begin{cases}
r^i_t=v^i r^i_x+p^i\eta^i_x  \\
\eta^i_t=v^i\eta^i_x  
\end{cases},
\end{equation}
requiring the Hamiltonian property in Dubrovin-Novikov sense. As a consequence of Tsarev's Theorem, we obtain a specific structure for the metric $g^{ij}$ and some necessary constraints on the interaction kernel. 

In particular, in  \cite{VerFer1,VerFer2} the authors   firstly proved that by conditions \eqref{cond1} the leading coefficient $g^{ij}$ has a block-diagonal form 
$$g^{ij}=\begin{pmatrix}
    J_1&0&\cdots &0\\
    0&J_2& \cdots &0\\
    0&0&\ddots&0\\
    0&0&\dots& J_n
\end{pmatrix}, \quad \text{where} \qquad J_i=\begin{pmatrix}
        m_i&n_i\\n_i&0
    \end{pmatrix}, \qquad i=1,2,\dots n.$$ We remark that each one of the $n$ blocks has a Hankel structure. 
    
Finally, applying conditions \eqref{cond2} we are also able to specialise the entries of $g^{ij}$:
\begin{equation}\label{nm}
n_i=\frac{s_i(\eta^i)}{(u^i)^2}, \quad m_i=-\frac{2s_i(\eta^i)}{(u^i)^3}\displaystyle \sum_{j\neq i}^n u^j\epsilon^{ji}_{,\eta^i} +\frac{g_i(r^i, \eta^i)}{(u^i)^2},
\end{equation}
where $u^i$ are defined in \eqref{uvp} and the functions  $s_i(\eta^i)$ and $g_i(r^i, \eta^i)$ are arbitrary.

The missing requirement to obtain the Hamiltonian structure is finally given by the flatness of the leading coefficient:

\begin{theorem}[\cite{VerFer2}]\label{main} The metric specified by (\ref{nm}) is flat if and only if the functions $g_i(r^i, \eta^i)$ are quadratic in $r^i$,
$$
g_i(r^i,\eta^i)=\varphi_i(\eta^i)(r^i)^2+\chi_i(\eta^i)r^i+\psi_i(\eta^i),
$$ 
furthermore, the following conditions must be satisfied:
\begin{subequations}
\begin{align}
&\epsilon^{ij}\left(\chi_i+\chi_j\right)=2\left(s_i\epsilon^{ij}_{,\eta^i}+s_j\epsilon^{ij}_{,\eta^j}+\displaystyle \sum_{k\neq i,j}\epsilon^{ik}\epsilon^{jk}\varphi_k\right), \label{eqdr}
\\
&\displaystyle \sum_{k\neq i} \varphi_{k}(\epsilon^{ik})^2+\psi_i=0. \label{eqref2}
\end{align}
\end{subequations}
\end{theorem}

This result shows that a Hamiltonian structure for the reduced systems with Dubrovin-Novikov operators is possible only under additional constraints on the interaction kernel $\epsilon$. In spite of the large number of examples that this assumption covers, for a number of weights $u^i$ larger than $n=2$, only a single Hamiltonian structure exists in the investigated cases (see Section \ref{sub33}). In particular, it seems that the bi-Hamiltonian  structure obtained with local operators in \cite{VerFer1} is lost for $n>2$.

However, as firstly stated in \cite{VerFer2}, the bi-Hamiltonian nature of the equations is preserved if we allow the operators to  be nonlocal.}

\section{Nonlocal Hamiltonian structures of the reduced systems}\label{sec0}
{As briefly shown in the previous Section,} we first remark that Tsarev's conditions \eqref{cond1} and \eqref{cond2} do not imply that a first-order operator \eqref{dn} is Hamiltonian. In particular, this result is not related to the curvature tensor of $g$. Indeed, the same conditions hold for a generalised version of Dubrovin-Novikov operators, which are nonlocal. 

Let us introduce the operator 
\begin{equation}\label{fm}
    g^{ij}(u)\partial_x+b^{ij}_k(u)+c\, u^i_x\,\partial_x^{-1}\,u^j_x, \qquad i,j=1,2,\dots n,
\end{equation}
where $c$ is a constant. These operators were firstly introduced in \cite{MF} by Ferapontov and Mokhov, who also proved the following result on their Hamiltonianity
\begin{theorem}[\cite{MF}]\label{fmthm}
    A nonlocal first-order operator of form \eqref{fm} is Hamiltonian if and only if setting $g_{ij}=(g^{kl})^{-1}$:
    \begin{itemize}
        \item $g_{ij}$ has constant curvature $c$; 
        \item $b^{ij}_k$ and the Christoffel symbols for $g$ are related by the formula\begin{equation}b^{ij}_k=-g^{is}\Gamma^j_{sk}, \qquad i,j,k=1,2,\dots n.
        \end{equation}
    \end{itemize}
\end{theorem}
We recall that the curvature tensor related to a semi-Riemannian metric $g_{ij}$ is
\begin{equation}\label{1o}
R^i_{jkl}=\Gamma^i_{jl,k}-\Gamma^i_{jk,l}+\Gamma^i_{ks}\Gamma^s_{jl}-\Gamma^i_{ls}\Gamma^s_{jk},
\end{equation}
where $\Gamma^{i}_{jk}$ are Christoffel symbols of the  Levi-Civita connection of $g$. So that we can re-write in coordinates the first point of Theorem \ref{fmthm} as
\begin{equation}\label{ccurv}R^i_{jkl}=c(\delta^i_kg_{jl}-\delta^i_lg_{jk}), \qquad i,j,k,l=1,2,\dots n.
\end{equation}
Notice that in the particular case of $c=0$ (i.e. the curvature tensor identically vanishes and the metric is flat) we recover the result of Theorem \ref{dnthm}.

Operators related to constant-curvature metrics can be viewed as particular cases of a further generalisation of Dubrovin-Novikov operators with a more complicated type of nonlocalities. This case has been investigated in full generality by Ferapontov in \cite{F} so that the following operator is also known as Ferapontov operator with $N$ nonlocal tails
\begin{equation}\label{cfop}
    g^{ij}(u)\partial_x+b^{ij}_k(u)\, u^k_x+ \displaystyle \sum_{\alpha=1}^N w^i_{\alpha k}(u)u^k_x\, \partial_x^{-1}\, w^j_{\alpha s}(u)\, u^s_x, \qquad i,j=1,2,\dots , n.
\end{equation}
The Hamiltonian property of \eqref{cfop} is described by the following Theorem:
\begin{theorem}[\cite{F}]
    A nonlocal first-order operator of form \eqref{cfop} is Hamiltonian if and only if setting $g_{ij}=(g^{kl})^{-1}$: 
    \begin{itemize}
        \item the pseudo-Riemannian metric $g_{ij}$ and the (1,1) tensor $w^i_j$ satisfy the relations
        \begin{subequations}\label{gpc}\begin{gather}\label{gpceqs}
            g_{is}w^s_{\alpha j}=g_{js}w^s_{\alpha i}, \quad \nabla_kw^i_{\alpha j}=\nabla_jw^i_{\alpha k}, \\  R^{ij}_{kl}=\displaystyle \sum_{\alpha=1}^N w^i_{\alpha k}w^j_{\alpha l} -w^j_{\alpha k}w^i_{\alpha l}, \qquad i,j,k,l=1,2,\dots n,
        \end{gather}
        \end{subequations}
        where $R^i_{jkl}=g_{js}R^{si}_{kl}$ is the Riemann curvature tensor of $g_{ij}$;
        \item the set of affinors $\{w^i_{\alpha j}\}_{\alpha=1}^N$ is commutative, i.e $[w_\alpha,w_\beta]=0$; 
        \item $b^{ij}_k$ and the Christoffel symbols for $g$ are related by the formula\begin{equation}b^{ij}_k=-g^{is}\Gamma^j_{sk}, \qquad i,j,k=1,2,\dots n.\end{equation}
    \end{itemize}
\end{theorem}
{\begin{remark}[A geometric interpretation of operators \eqref{cfop}]
     Ferapontov described a beautiful geometric interpretation of nonlocal operators \eqref{cfop} in the differential-geometric context. In \cite{F}, equations \eqref{gpc} are interpreted as the Gauss-Peterson-Codazzi equations for submanifolds of dimension $n$ embedded in the pseudo-Euclidean space $\mathbb{R}^{n+N}$ with flat normal connection. In particular, the operators $w^i_{\alpha j}$ are regarded as the shape operators of the submanifold $M$ (also known as the Weingarten operators) corresponding to the field of pairwise orthogonal unit normals $\bar{n}_\alpha$ and the metric $g_{ij}$ as the first fundamental form. Notice that the family of shape operators is commutative by definition of submanifold with flat normal connection. 
     
     In the simple case $N=1$, we recall that if we indicate with $\mathbb{I}$ the second fundamental form of the hypersurface $M$, the shape operator satisfies
     \begin{equation}
         w^i_j=g^{is}\mathbb{I}_{sj}, \qquad i,j=1,2,\dots n,
     \end{equation}
     so that given a Hamiltonian operator in form \eqref{cfop} we uniquely determine a hypersurface in $\mathbb{R}^{n+1}$ in terms of its two fundamental forms.

     We stress that for flat metrics $g_{ij}$ (where $R^i_{jkl}=0$), the hypersurface reduces to a hyperplane and the Weingarten operator vanishes identically. So that, this case covers the geometric interpretation of Dubrovin-Novikov operators, which indeed are local. \end{remark}}

The compatibility conditions found by Tsarev have been recently computed for nonlocal operators by Vitolo and the present author in \cite{VerVit1}:
\begin{theorem}[\cite{VerVit1}]\label{thm_vervit2}
   Let us consider a non-local first order Hamiltonian operator of form \eqref{cfop} with $N=1$, whose
nonlocal part is defined by a hydrodynamic type symmetry $\varphi^i=w^i_j(u)u^j_x$,
and the hydrodynamic type system \eqref{hts}. Then, if the system is Hamiltonian with a Ferapontov operator the following conditions must be satisfied:
  \begin{subequations}\begin{align}
        &g^{is}V^j_s=g^{js}V^i_s, \qquad i,j=1,2,\dots n,\label{cond111}\\
        &\nabla^iV^j_k=\nabla^jV^i_k, \qquad i,j,k=1,2,\dots n.\label{cond211}
    \end{align}
    \end{subequations}
\end{theorem}

This reveals that the necessary conditions for local \eqref{dn} or nonlocal \eqref{cfop} operators to be compatible with hydrodynamic-type systems do not change provided that the affinor $w^i_j$ is a symmetry for the system. We remark that the case $w^i_j=c\delta^i_j$ reduces to the constant-curvature operator and the additional requirement for $w^i_j$ to be a symmetry for the system is trivially satisfied by systems $u^i_t=V^i_j(u)u^j_x$ not explicitly depending on the independent variable $x$ (see also \cite[Section 4]{Ver1} and discussion therein). 

We finally briefly refer to a particular nonlocal Hamiltonian structure given by the following operator
\begin{equation}\label{confop}
    g^{ij}(u)\partial_x+b^{ij}_k(u)u^k_x+w^i_s(u)u^s_x\partial_x^{-1}u^j_x+u^i_x\partial_x^{-1}w^j_s(u)u^s_x, \qquad i,j=1,2,\dots, n.
\end{equation}

This operator has two nonlocal tails and is strictly related to conformally flat metrics $g_{ij}$, i.e. metrics such that $g_{ij}=\Omega(u)\, \eta_{ij}$ and $\eta_{ij}$ is a Euclidean metric. The Hamiltonianity conditions for \eqref{confop} have been computed again by Ferapontov in the very short paper \cite{Fer34}:
\begin{theorem}[\cite{Fer34}]
    A nonlocal first-order operator of form \eqref{confop} is Hamiltonian if and only if 
    setting $g_{ij}=(g^{kl})^{-1}$: 
    \begin{itemize}
        \item the pseudo-Riemannian metric $g_{ij}$ and the (1,1) tensor $w^i_j$ satisfy the relations
        \begin{subequations}\label{gpcr}\begin{gather}
            g_{is}w^s_{ j}=g_{js}w^s_{ i}, \quad \nabla_kw^i_{ j}=\nabla_jw^i_{ k}, \\  R^{ij}_{kl}=w^i_k\delta^j_l+w^j_l\delta^i_k-w^j_k\delta^i_l-w^i_l\delta^j_k., \qquad i,j,k,l=1,2,\dots n,
        \end{gather}
        \end{subequations}
        where $R^i_{jkl}=g_{js}R^{si}_{kl}$ is the Riemann curvature tensor of $g_{ij}$;
        \item $b^{ij}_k$ and the Christoffel symbols for $g$ are related by the formula\begin{equation}b^{ij}_k=-g^{is}\Gamma^j_{sk}, \qquad i,j,k=1,2,\dots n.\end{equation}
    \end{itemize}
\end{theorem}

As remarked in \cite{Fer34}, from a purely differential geometric viewpoint the operator \eqref{confop} is Hamiltonian if and only if $g_{ij}$ is a conformally flat metric.

\vspace{3mm}

{At this point, starting from Tsarev's compatibility relations between a first-order operator and a system of hydrodynamic type, we might require one of the following additional conditions
\begin{itemize}
    \item[(a)] $g$ is a flat metric, 
    \item[(b)] $g$ is a constant-curvature metric, or
    \item[(c)] $g$ is a conformally flat metric. 
\end{itemize}

These requirements can also be further generalised, see \cite{F}.

Condition (a) was the one discussed in \cite{VerFer2}, the main result in this case is presented in Theorem \ref{main} of the previous Section. In this paper, we address our investigation to cases (b) and (c).}

\subsection{Hamiltonian structures with constant-curvature metric}

We now wonder if another Hamiltonian structure with constant curvature emerges for systems of type \eqref{J2}. This turns out to be the case. In particular, we are able to extend conditions \eqref{eqdr} and \eqref{eqref2} involving an arbitrary constant $c$ related to the curvature of the proposed operator. {We remark that the following statement first appeared in \cite{VerFer2} without a proof. In this paper, we show the proof in details and present some further results as direct consequence of the Theorem. The Corollaries here obtained cover a large class of interaction kernels, i.e. the separable ones.}

The following result holds true:

\begin{theorem}\label{main2} The following conditions are necessary for the metric $g$ to {have} constant curvature $c$:
\begin{subequations}\begin{align}
&\epsilon^{ij}\left(\chi_i+\chi_j\right)-2c=2\left(s_i\epsilon^{ij}_{,i}+s_j\epsilon^{ij}_{,j}+\displaystyle \sum_{k\neq i,j}\epsilon^{ik}\epsilon^{jk}\varphi_k\right)\label{eqdr1}
\\&
\displaystyle \sum_{k\neq i} \varphi_{k}(\epsilon^{ik})^2+\psi_i=-c\label{eqref21}
\end{align}\end{subequations}
where 
$g_i(r^i,\eta^i)=\varphi_i(\eta^i)(r^i)^2+\chi_i(\eta^i)r^i+\psi_i(\eta^i)
$.\end{theorem}
\begin{proof}
The proof of the Theorem follows from considering the component $R^{r^i}_{r^ir^i\eta^i}$ of the Riemann curvature tensor. Applying the definition of constant-curvature Riemann tensor \eqref{ccurv}, the following condition must be satisfied
\begin{equation}\label{33}
R^{r^i}_{r^ir^i\eta^i}=-cg_{r^i\eta^i},
\end{equation} 
for an arbitrary constant $c$. 
Here, the left-hand side is
\begin{equation}\label{18}
R^{r^i}_{r^ir^i\eta^i}=\Gamma^{r^i}_{r^i\eta^i,r^i}-\Gamma^{r^i}_{r^ir^i,\eta^i}+\Gamma^{r^i}_{r^is}\Gamma^s_{r^i\eta^i}-\Gamma^{r^i}_{\eta^is}\Gamma^s_{r^ir^i}.
\end{equation}
One can now substitute the expression of the Christoffel symbols (see  \cite[Theorem 1]{VerFer2}) to make \eqref{18} explicit, as 

\begin{equation}
    - \frac{(\det \epsilon)^2g_{i,r^ir^i}-2(\det\epsilon )A_{i,i}g_{i,r^i}+\displaystyle 2\sum_{k=1}^ng_k(A_{i,k})^2{+}2\sum_{k,l=1}^nA_{i,k}A_{i,l}\left(s_l\epsilon^{lk}_{,l}+s_k\epsilon^{lk}_{,k}\right)}{2\, s_i\, (\det \epsilon)^2}
\end{equation}
where $A_{i, k}$ is the cofactor of the $n\times n$ matrix $\hat \epsilon$,  i.e.,  the determinant of the minor obtained by eliminating the $i$-th row and the $k$-th column of $\hat \epsilon$.

Moreover, using \eqref{nm} we obtain that the right-hand side of the previous equation is \begin{equation}
\frac{c(u^i)^2}{s_i(\eta^i)},
\end{equation}
and recalling that $u^i=\sum_k\beta_{ik}$, where $\beta_{ik}=(-1)^{i+k+1}\frac{A_{i,k}}{\det \epsilon}$. In particular,  
$$c(u^i)^2 (\det \epsilon)^2=c\displaystyle \left(\sum_{k=1}^n(-1)^{i+k+1}A_{i,k}\right)^2$$
Finally, one can easily see as \cite[Theorem 1]{VerFer2} that \eqref{33} does not depend on $r^i$.  Then,  $g_i$ is quadratic in $r^i$.

To conclude, conditions \eqref{eqdr1}, \eqref{eqref21} follow considering respectively the coefficients of $(r^j)\prod_{k\neq i,j} (r^k)^2$ and $\prod_{k\neq i}(r^k)^2$.
\end{proof}

\vspace{3mm}

In addition, we easily obtain the following Corollaries:
\begin{corollary}By the previous conditions it follows that two cases arise:
\begin{itemize}
\item[A)] The kernel $\epsilon$ is separable;
\item[B)] $\varphi_i=0,\psi_i=-c$ and equation \eqref{eqdr1} is satisfied;
\end{itemize}\end{corollary}
\begin{proof}
    Let us consider equation \eqref{eqref21}:
\begin{equation}
    \displaystyle \sum_{k\neq i} \varphi_{k}(\epsilon^{ik})^2+\psi_i=c,
\end{equation}
recalling that $\varphi_i=\varphi(\eta^i)$ and $\psi_j=\psi_j(\eta^j)$. Differentiating twice with respect to $\eta^k$ (with $k\neq  i$), dividing the result by $(\epsilon^{ik})^2$ and then differentiating again by $\eta^i$, we have 
\begin{equation}
    \varphi_k\frac{\partial^2 (\log \epsilon^{ik})}{\partial \eta^i\partial \eta^k}=0.
\end{equation}

We then obtain two cases
\begin{equation}
    \varphi_k=0, \qquad \text{or} \qquad \frac{\partial^2 (\log \epsilon^{ik})}{\partial \eta^i\partial \eta^k}=0.
\end{equation}
The first equation simply leads to $\psi_k=-c$, jointly with condition \eqref{eqdr1}. Then, case B) is proved. 

Solving the second one, we obtain 
\begin{equation}
    \frac{\partial^2 \epsilon^{ik}}{\partial \eta^i\partial \eta^k}=\frac{1}{\epsilon^{ik}}\frac{\partial \epsilon^{ik}}{\partial \eta^i}\, \frac{\partial \epsilon^{ik}}{\partial \eta^k},
\end{equation} that is the interaction kernel $\epsilon(\eta^i,\eta^k)$ is multiplicatively separable:
\begin{equation}
    \epsilon(\eta^i,\eta^k)=\phi_i(\eta^i)\phi_k(\eta^k),
\end{equation}
that is, case A) is also proved.
\end{proof}

\vspace{3mm}

\begin{corollary}
    In the separable case, we additionally have that \begin{equation}
        \epsilon^{ij}(\eta^i,\eta^j)=\phi_i(\eta^i), \qquad \text{or}\qquad \epsilon^{ij}(\eta^i,\eta^j)=\phi_j(\eta^j)
    \end{equation}
\end{corollary}
\begin{proof}
    Let us suppose $\epsilon^{ij}(\eta^i,\eta^j)=\phi_i(\eta^i)\phi_j(\eta^j)$, then equation \eqref{eqdr1} reduces to 
    \begin{equation}
        \phi_i\, \phi_j \left(\chi_i+\chi_j\right) -2c = 2s_j\,  \phi_i\, \phi_j'+ 2s_i\, \phi_i'\, \phi_j.
    \end{equation}
    Let us now divide the previous expression by $\phi_i\, \phi_j$ and obtain the following
    \begin{equation}\label{questaqui}
        \chi_i+\chi_j-\frac{2c}{\phi_i\phi_j} = 2s_j\, \frac{\phi_j'}{\phi_j}+2s_i\, \frac{\phi_i'}{\phi_i}.
    \end{equation}
    Now by recalling that $s_\ell \frac{\phi_\ell'}{\phi_\ell}$ only depends on $\eta^\ell$, then applying$\dfrac{\partial^2}{\partial \eta^i\partial \eta^j}$ we reduce the equation \eqref{questaqui} into 
    \begin{equation}
        -2c\frac{\phi_i'\, \phi_j'}{\phi_i^2\phi_j^2}=0.
    \end{equation}
    So that the Corollary is proved. 
\end{proof}

\subsection{Conformally flat metrics and Ferapontov Hamiltonian structures}\label{subconf}
We conclude this Section with some further remarks on general Ferapontov structures \eqref{cfop} and structures related to conformally flat metrics \eqref{confop} for systems in Jordan block form. Let us firstly stress that for $N=1$ the necessary condition of compatibility for a system to be Hamiltonian with Ferapontov operators are given by Theorem \ref{thm_vervit2}. In particular, we need the additional requirement that the affinor $w^i_j(u)$ defines a hydrodynamic-type symmetry $\varphi^i=w^i_j(u)u^j_x$ for the system, i.e. they must commute. An {\color{blue}analogous} result holds true for operators \eqref{confop}.

In \cite{FP}, the authors proved that commuting flows for polychromatic reductions of the El's kinetic equations are given by affinors in the same Jordan block form, where the blocks are in Toeplitz form: 
\begin{equation}
w=\begin{pmatrix}
    A_1&0&\cdots &0\\
    0&A_2& \cdots &0\\
    0&0&\ddots&0\\
    0&0&\dots& A_n
\end{pmatrix}, \quad \text{ } \qquad A_i=\begin{pmatrix}
        w^i&q^i\\0&w^i
    \end{pmatrix}, \qquad i=1,2,\dots n.    
\end{equation}
and  
\begin{align}\label{1a}
w^i&=\frac{1}{u^i}\beta_{ki}\varphi^i,\\
q^i&=\frac{1}{u^i}\left(\epsilon^{ki}_{,\eta^i}(w^k-w^i)-\mu^ir^i+\varphi^i_{\eta^i}\right),\label{2a}
\end{align}
where $\varphi^i(\eta^1,\eta^2)$ and $\mu^i(\eta^i)$, $i=1,2,\dots,  n$ are arbitrary functions of the indicated arguments with the additional relation $\partial_i\varphi^j=\epsilon^{ij}\mu^i$.

As an example, for $n=2$ the affinor must be
\begin{equation}
w=\begin{pmatrix}
w^1&p^1&&\\&w^1&&\\&&w^2&p^2\\&&&w^2
\end{pmatrix}
\end{equation}
where explicitly (see \cite[formula (42)]{FP}):
\begin{subequations}
\begin{align}
    &w^1=\frac{r^2\varphi^1-\epsilon \varphi^2}{r^2-\epsilon},\\
    &w^2=\frac{r^1\varphi^2-\epsilon \varphi^1}{r^1-\epsilon},\\
    &q^1=\frac{r^1r^2-\epsilon^2}{r^2-\epsilon}\left(\frac{\varphi^2-\varphi^1}{r^2-\epsilon}\epsilon_{,\eta^1}+r^1\mu^1-\varphi^1_{,\eta^1}\right),\\
    &q^2=\frac{r^1r^2-\epsilon^2}{r^1-\epsilon}\left(\frac{\varphi^1-\varphi^2}{r^1-\epsilon}\epsilon_{,\eta^2}+r^2\mu^2-\varphi^2_{,\eta^2}\right).
\end{align}
\end{subequations}
Here $\varphi^i(\eta^1,\eta^2)$, $i=1,2$, are arbitrary functions satisfying $\varphi^1_{,\eta^2}=\epsilon\, \mu^2$ and $\varphi^2_{,\eta^1}=\epsilon\, \mu^1$, whereas $\mu^i(\eta^i)$, $i=1,2$ are totally arbitrary functions of their arguments.

We finally stress that deriving the analogue of Theorems \ref{main} and \ref{main2} is a harder task. Indeed, in this case the Hamiltonianity conditions deal with a different type of curvature tensor, i.e. a tensor of type (2,2). This fact implies the need to {\color{blue} raise} one index in the Riemann tensor, producing a more complicated expression to be solved. However, for concrete computations we skip this problem by solving the additional requirement that the affinor $w^i_j$ in the nonlocal tail of the operators must be a symmetry of the investigated system.

In the following Section, we show some examples coming form the theory of soliton gas and generalised hydrodynamics.


\section{Examples}\label{sub33}
\begin{remark}
    Before proceeding with some examples, we briefly discuss the structure of the Hamiltonian density $H=\int{h(u)\, dx}$ for the reduced systems. We firstly remark that, as in the finite dimensional case, the Hamiltonian functional is a conserved quantity for the related Hamiltonian system. The structure of conservation laws for polychromatic reductions of the kinetic equation was established in \cite{FP}. In our case, the Hamiltonian density $h$ is given by the following formula
\begin{equation}\label{hh} h({r},\eta)=\displaystyle \sum_{i=1}^nu^i\, h_i(\eta^i)
\end{equation}
where $u^i$ are defined in \eqref{uvp} and $h_i(\eta^i)$ must be specified case by case when fixing the interaction kernel $G(\mu,\eta)$ and the effective velocities $v$.
\end{remark}
The El's system \eqref{gas} can be derived as a thermodynamic limit of the Whitham equations for a large class of integrable equations. El himself proved in \cite{El} that the Korteweg-de Vries equation is associated with \eqref{gas} for the specific choice of the interaction kernel $G(\mu,\eta)$ and the free-velocity $S(\eta)$ as follows
\begin{equation}
    G(\mu,\eta)=\frac{1}{\mu\eta}\log{\left|\frac{\mu-\eta}{\mu+\eta}\right|}, \qquad S(\eta)=4\eta^2.
\end{equation}

In the following table we list other cases, for each one we show the interaction kernel and the free-velocity functions.

\vspace{3mm}

\begin{center}
 \footnotesize{Table 1. Types of soliton gas equations}
 
 \vspace{2mm}
\begin{tabular}{|l|l|l|}
 \hline
kinetic equation & $S(\eta)$  & $G(\mu, \eta)$ \\
 \hline
KdV soliton gas\vphantom{$\frac{\frac{A}A}{\frac{A}A}$} &
$4\eta^2$ & 
$\frac{1}{\eta \mu} \log \bigg|{ \frac{\eta-\mu}{\eta+\mu}}\bigg|$\\
 \hline
sinh-Gordon soliton gas\vphantom{$\frac{\frac{A}A}{\frac{A}A}$} &
$\tanh \eta$ &  
$\frac{1}{\cosh \eta \cosh  \mu}\  \frac{a^2\cosh(\eta-\mu)}{4\sinh^2(\eta-\mu)}$\\
\hline
hard-rod gas\vphantom{$\frac{\frac{A}A}{\frac{A}A}$} &
$\eta$ &  
$-a$\\
\hline
Lieb-Liniger gas\vphantom{$\frac{\frac{A}A}{\frac{A}A}$} &
$\eta$ &  
$\frac{2a}{a^2+(\eta -\mu)^2}$\\
\hline
DNLS soliton gas\vphantom{$\frac{\frac{A}A}{\frac{A}A}$} &
$\eta$ &  
$\frac{1}{2\sqrt{\eta^2-1}\sqrt{\mu^2-1}}\log  \left( \frac{(\eta-\mu)^2-\left(\sqrt{\eta^2-1}+\sqrt{\mu^2-1}\right)^2}{(\eta-\mu)^2-\left(\sqrt{\eta^2-1}-\sqrt{\mu^2-1}\right)^2}\right)$  \\
\hline
separable case\vphantom{$\frac{\frac{A}A}{\frac{A}A}$} &
arbitrary &  
$\phi(\eta)+\phi(\mu)$\\
\hline
general case\vphantom{$\frac{\frac{A}A}{\frac{A}A}$} &
arbitrary &  
$\phi(\mu)\phi(\eta)g[a(\mu)-a(\eta)]$\\
\hline
\end{tabular}

\end{center}

\vspace{3mm}

The existence of a local Hamiltonian structure for the previous cases was proved in \cite{VerFer2}. Here we list their results by making explicit the functions $s_i$ and $\chi_i$ in Theorem \ref{main} {and the Hamiltonian density $h_i$ as in \eqref{hh}:}

 \begin{center}
 \footnotesize{Table 2: Local Hamiltonian structures for $n\geq 3$}
 
 \vspace{4mm}
 
\begin{tabular}{|l|l|l|l|}
 \hline
kinetic equation & $s_i(\eta^i)$  & $\chi_i(\eta^i)$ &$h_i(\eta^i)$\\
 \hline
KdV soliton gas\vphantom{$\frac{\frac{A}A}{\frac{A}A}$} &
$\eta^i$ & 
$-2$ & $-\frac{4}{3}(\eta^i)^2$\\

 \hline
sinh-Gordon soliton gas\vphantom{$\frac{\frac{A}A}{\frac{A}A}$} &
$1$ & 
$-2\tanh \eta^i$ & $-1$\\

\hline

Lieb-Liniger gas\vphantom{$\frac{\frac{A}A}{\frac{A}A}$} &
$1$ & 
$0$ & $-\frac{1}{2}(\eta^i)^2$\\

\hline

DNLS soliton gas\vphantom{$\frac{\frac{A}A}{\frac{A}A}$} &
$1-(\eta^i)^2$ & 
$2\eta^i$ & $1$\\

\hline
separable case\vphantom{$\frac{\frac{A}A}{\frac{A}A}$} 
& $\frac{\phi(\eta^i)}{\phi'(\eta^i)} $ &  
$1$ & $-\sqrt{\phi(\eta^i)}\int^{\eta^i}\frac{\phi'(\eta)S(\eta)}{\phi(\eta)^{3/2}}\ d\eta $\\

\hline

general case\vphantom{$\frac{\frac{A}A}{\frac{A}A}$} &
$\frac{1}{a'(\eta^i)}$ & 
$\frac{2\phi'(\eta^i)}{a'(\eta^i)\phi(\eta^i)}$ & $-\phi(\eta^i)\int^{\eta^i}{\frac{S(\eta)a'(\eta)}{\phi(\eta)}\, d\eta}$\\

 \hline

\end{tabular}
 \end{center}

 \vspace{3mm}

 \paragraph{Hamiltonian structures with constant-curvature metrics}

 We now investigate the existence of a nonlocal structure with the Hamiltonian operator defined in \eqref{fm}. We show two examples of this type and finally a case not admitting such structure.

\begin{example}[KdV equation - I]\label{kdvexa}
Let us firstly consider the $n=2$ case, i.e. 
\begin{equation}
\epsilon(\eta^1,\eta^2)=\frac{1}{\eta^1\eta^2}\ln{\left(\frac{\eta^1-\eta^2}{\eta^1+\eta^2}\right)}.
\end{equation} 
To search for a nonlocal Hamiltonian structure of type \eqref{fm} one can proceed in two ways:
\begin{itemize}
    \item[1.] solve condition \eqref{ccurv} in the unknown functions $s_1,s_2$, $\psi_1,\psi_2$ and $\chi_1,\chi_2$, or
    \item[2.] solve \eqref{eqdr1} in the same unknown functions but double checking that the solutions give arise to a constant-curvature metric (indeed, we recall that condition \eqref{eqdr1} is only necessary and not sufficient in general).
\end{itemize}

We obtain the following
\begin{subequations}
\begin{gather}
s_1(\eta^1)=-\frac{(c_1+c_2)\eta^1}{4}-\frac{c}{2}(\eta^1)^3,\quad s_2(\eta^2)=-\frac{(c_1+c_2)\eta^2}{4}-\frac{c}{2}(\eta^2)^3,\\\psi_1(\eta^1)=-c,\quad \psi_2(\eta^2)=-c, \\  \chi_1(\eta^1)=c_1+c(\eta^1)^2,\quad \chi_2(\eta^2)=c_2+c(\eta^2)^2
\end{gather}
\end{subequations}

We remark that in the 2-dimensional reduction we obtain a multi-Hamiltonian structure, with three Hamiltonian operators one for each arbitrary constant $c_1,c_2$ and $c$. Two of them are local and one is nonlocal. 

Finally, we generalise the previous result for arbitrary $n$, having 
\begin{subequations}
\begin{gather}
s_i(\eta^i)=\frac{c}{2}(\eta^i)^3-\frac{\tilde{c}}{2}\eta^i, \\ \psi_i(\eta^i)=c, \quad \chi_i(\eta^i)= \tilde{c}-c(\eta^1)^2.
\end{gather}
\end{subequations}
This shows that for higher order reductions only one local structure survives (with respect to $c_1=c_2=-\tilde{c}$) but the bi-Hamiltonian structure is preserved by the nonlocal one.

Finally the Hamiltonian density is given by formula \eqref{hh}, where we specify
\begin{equation}
    h_i=-\frac{8}{\eta^i}\displaystyle \int^{\eta^i}{ \frac{\eta^2}{\tilde{c}-c\eta^2}\, d\eta}=\begin{cases}-\dfrac{8}{c}+\dfrac{8\sqrt{2}\tan^{-1}{\left(\sqrt{\frac{c}{\tilde{c}}}\eta^i\right)}}{\eta^ic\sqrt{c\cdot \tilde{c}}}\qquad \tilde{c},c\neq 0\\[1.5ex]
    -\dfrac{4}{3\tilde{c}}(\eta^i)^2\qquad \qquad \qquad \qquad c=0\\[2ex]
    -\dfrac{8}{c}\qquad \qquad \qquad \qquad \quad \quad \,\, \,  \tilde{c}=0
    \end{cases}
\end{equation}
\end{example}

Analogously, we consider a second example. 
\begin{example}[Additive separable cases]
Let us consider an interaction kernel of the form
\begin{equation}
    \epsilon(\eta^i,\eta^j)=\phi_i(\eta^i)+\phi_j(\eta^j),
\end{equation}
also known as additive separable case.

Assuming $n=2$, we have again a multi-Hamiltonian structure:
\begin{subequations}
\begin{gather}
s_1(\eta^1)=\frac{c_3\phi_1^2+2c_1\phi_1+c_4-c}{2\phi_1'},\quad s_2(\eta^2)=\frac{-c_3\phi_2^2+2c_1\phi_2-c_4-c}{2\phi_2'},\\ 
\psi_1(\eta^1)=-c,\quad \psi_2(\eta^2)=-c, \\ \chi_1(\eta^1)=c_1+c_2+c_3\phi_1,\quad \chi_2(\eta^2)=c_1-c_2-c_3\phi_2
\end{gather}
\end{subequations}
and finally generalising for arbitrary $n$, we obtain that 
\begin{align}
s_i(\eta^i)=\frac{2\phi_i(\eta^i)\tilde{c}-c}{2\phi_i'(\eta^i)}, \qquad \psi_i(\eta^i)=-c, \qquad \chi_i(\eta^i)= \tilde{c}.
\end{align}
The quantities $h_i$ in formula \eqref{hh} are
\begin{equation}
    h_i=2\sqrt{2\phi(\eta^i)+c}\displaystyle \int^{\eta^i}{\frac{\phi'(\eta)S(\eta)}{\left(2\phi(\eta)+c\right)^{3/2}} d\eta}
\end{equation}
so that also the Hamiltonian density is established.
\end{example}

We finally consider an example for which such a nonlocal structure does not exist.  

\begin{example}[Lieb-Liniger] 
Let us consider the interaction kernel 
\begin{equation}\label{lls}
    \epsilon(\eta^i,\eta^j)=\frac{2a}{a^2+(\eta^i-\eta^j)^2},
\end{equation}where $a$ is a constant. This phase shift is associated with the Lieb-Liniger equation. It is easy to check that by substituting \eqref{lls} equation \eqref{eqdr1} has no solution with arbitrary $c$ and the only one obtained requires $c=0$. As a result, the Hamiltonian structure reduces to a purely local one, i.e. to a Dubrovin-Novikov operator.  \end{example}

\vspace{3mm}

\paragraph{Hamiltonian structures related to conformally flat metrics} Finally, we consider operators of type \eqref{confop}. We recall that in this case the affinor $w^i_j$ is required to be a hydrodynamic-type symmetry for the reduced system, so that computing the Hamiltonianity conditions for operators of the form described in subsection \ref{subconf} we can compute two main examples. 

\begin{example}[KdV equation - II]
Fixing the interaction kernel of the KdV, solving the Hamiltonianity conditions of the operator and considering equations \eqref{1a} and \eqref{2a} with the additional relation $\partial_i\varphi^j=\epsilon^{ij}\mu^i$, we obtain the following
\begin{subequations}
\begin{align}
s_1(\eta^1)&=-\frac{1}{4}\left(2c_2(\eta^1)^4+2c_1(\eta^1)^2+2c_3(\eta^1)^2+c_4+c_5\right)\eta^1,\\
s_2(\eta^2)&=-\frac{1}{4}\left(2c_2(\eta^2)^4+2c_1(\eta^2)^2+2c_3(\eta^2)^2+c_4+c_5\right)\eta^2,\\
g_1(r^1,\eta^1)&=-2c_2(\eta^1)^2+\left(c_2(\eta^1)^4+(c_1+c_3)(\eta^1)^2+c_4\right)r^1-2c_3,\\
g_2(r^2,\eta^2)&=-2c_2(\eta^2)^2+\left(c_2(\eta^2)^4+(c_1+c_3)(\eta^2)^2+c_5\right)r^2-2c_3,
\end{align}\end{subequations}
whereas, the flow $w^i_j$ has the following {form}
\begin{subequations}
\begin{gather}
\varphi^1(\eta^1,\eta^2)=c_2(\eta^1)^2+c_3,\qquad 
\varphi^2(\eta^1,\eta^2)=c_2(\eta^2)^2+c_1,\\
\mu_1(\eta^1)=0,\qquad 
\mu_2(\eta^2)=0,
\end{gather}
\end{subequations}
Note that we reduce to the flat case when $w^i_j=0$. This is equivalent to require that $c_1=c_2=c_3=0$. In this case, we obtain the same solution obtained in \cite{VerFer1,VerFer2}. Furthermore, the constant curvature case of Example \ref{kdvexa} is derived by choosing $c_1=c_3$ and $c_2=0$.

\end{example}

\begin{example}[Lieb-Liniger equation - II]
Choosing the interaction kernel of the Lieb-Liniger model, we obtain:
\begin{subequations}
\begin{gather}
s_1(\eta^1)=c_3\qquad 
s_2(\eta^2)=c_3\\
g_1(r^1,\eta^1)=c_2r^1+2c_1\qquad 
g_2(r^2,\eta^2)=-(c_2r^2+2c_1)
\end{gather}\end{subequations}
whereas, the flow {$w^i_j$} has the following {form}
\begin{subequations}
\begin{gather}
\phi_1(\eta^1,\eta^2)=-c_1\qquad 
\phi_2(\eta^1,\eta^2)=c_1\\
\mu_1(\eta^1)=0\qquad 
\mu_2(\eta^2)=0
\end{gather}
\end{subequations}
Now we notice that  for $c_1=0$ the flat case is obtained. However, the constant curvature case is only obtained for $\phi_1(\eta^1,\eta^2)=\phi_2(\eta^1,\eta^2)$, i.e. for $c=c_3=0$. 
\end{example}

As expected, both the examples reveal that nonlocal operators allow a more general structure for the leading coefficient $g^{ij}(u)$. In particular, it is easy to check that the functions $g_i(r^i,\eta^i)$ are not only linear in $r^i$ for non-separable kernels. Moreover, the functions $s_i(\eta^i)$ in the KdV case and in the Lieb-Liniger one are polynomials of higher order compared to what obtained for flat structures. 


\section{Conclusions}\label{conc}
{In this article, we presented in details additional Hamiltonian structures for the polychromatic reduction of the kinetic equation for soliton gas. We recall that for reductions in higher number of components ($n>2$, i.e. more than 4 components) the structure obtained turns out to play a key role in the integrability of the system. Even if similar results  firstly appeared in another previous paper by E.V. Ferapontov and the present author, here we give some further details: we present a rigorous proof for operators with constant-curvature metrics,  we discuss the conformally flat case and compute new examples. }

The existence of a second Hamiltonian structure which is nonlocal reveals that a deeper investigation of the inner geometric properties of such equations is still needed. As an example, we conjecture (and some preliminary results jointly with E.V. Ferapontov confirm this thesis) that a more general Hamiltonian formulation is possible involving other type of nonlocalities, such as the one given by fully general Ferapontov operators \eqref{cfop}

In this direction, we expect to additionally investigate the geometric interpretation of the related manifolds in a purely differential geometric context. Finally, in the thermodynamic limit we plan to extend the Hamiltonian property with these nonlocalities to the full kinetic equation, as firstly done in \cite{VerFer2}.

Recently, new developments on degenerate operators \cite{Ver2,DellAVer} show that similar results can be obtained for Dubrovin-Novikov operators  without the assumption of non-degeneracy of the leading coefficient $g^{ij}$. We wonder if degenerate structures for the investigated equations exist or rather if such a generalisation might be helpful for other types of reductions (as the one recently obtained by T. Congy, M.A. Hoefer and G.A. El under the condensate ansatz) to prove their integrability.

\vspace{5mm}



\noindent \textbf{Acknowledgments}

{The author is extremely thankful to E.V. Ferapontov for the possibility to work on this problem, for his comments and suggestions and for his concrete help. He acknowledges the financial support of GNFM of the Istituto Nazionale di Alta Matematica.  The author is partially funded by the research project Mathematical Methods in Non- Linear Physics (MMNLP) by the Commissione Scientifica Nazionale – Gruppo 4 – Fisica Teorica of the Istituto Nazionale di Fisica Nucleare (INFN) and by the project “An artificial intelligence approach for risk assessment and prevention of low back pain: towards precision spine care”, PNRR-MAD-2022-12376692, CUP: J43C22001510001 funded by the European Union - Next Generation EU - NRRP M6C2 - Investment 2.1 Enhancement and strengthening of biomedical research in the NHS.}



\end{document}